\documentclass[11pt]{article}
\usepackage[utf8]{inputenc}
\usepackage{amsmath,amsfonts,amssymb}
\usepackage{amsthm}
\usepackage{latexsym}
\usepackage[noadjust]{cite}
\usepackage{graphicx}
\usepackage{xcolor}
\usepackage{dirtytalk}
\usepackage{mathtools}
\usepackage[margin=1.1in]{geometry}
\usepackage{thm-restate}
\usepackage{enumitem}
\usepackage[linesnumbered,ruled]{algorithm2e}
\usepackage[colorlinks=true, allcolors=blue]{hyperref}
\usepackage[nameinlink,capitalise]{cleveref}
\hypersetup{
    citecolor={violet}
}

\newtheorem{theorem}{Theorem}[section]
\newtheorem{corollary}[theorem]{Corollary}
\newtheorem{lemma}[theorem]{Lemma}

\newtheorem{claim}[theorem]{Claim}
\theoremstyle{definition}
\newtheorem{definition}[theorem]{Definition}
\newtheorem{remark}[theorem]{Remark}
\newtheorem{fact}[theorem]{Fact}

\newtheorem{conjecture}[theorem]{Conjecture}

\newcommand{\R}{\mathbb{R}}
\newcommand{\Z}{\mathbb{Z}}

\newcommand{\zo}{\{0, 1\}}

\newcommand{\eps}{\epsilon}

\newcommand{\CL}{\mathrm{CL}}

\newcommand{\CSP}{\textrm{CSP}}

\newcommand{\NRD}{\operatorname{NRD}}

\DeclareMathOperator{\z}{0}

\newif\ifdraft

\drafttrue

\title{Separating Non-redundancy and Chain Length}
\author{Joshua Brakensiek\thanks{University of California, Berkeley. Supported in part by a Simons Investigator award of Venkatesan Guruswami, and NSF awards CCF-2211972 and DMS-2503280. Contact: \href{mailto:josh.brakensiek@berkeley.edu}{josh.brakensiek@berkeley.edu}} \and Venkatesan Guruswami\thanks{Simons Institute for the Theory of Computing and the University of California, Berkeley. Supported in part by a Simons Investigator award and NSF award CCF-2211972. Contact: \href{mailto:venkatg@berkeley.edu}{venkatg@berkeley.edu}} \and Aaron Putterman\thanks{Harvard University. Supported in part by a Jane Street Graduate Research Fellowship, the Simons Investigator awards of Madhu Sudan and Salil Vadhan, and AFOSR award FA9550-25-1-0112. Contact: \href{mailto:aputterman@g.harvard.edu}{aputterman@g.harvard.edu}}}
\date{}

\begin{document}

\maketitle

\begin{abstract}
For a constraint satisfaction problem defined by a relation $R$, its non-redundancy $\NRD(R,n)$ is the size of largest instance (as a function of the number $n$ of variables) for which no constraint is implied by the rest. Its chain length $\CL(R,n)$ is the largest such instance where the constraints can be ordered so that no constraint is implied by the preceding ones. Clearly $\CL(R,n) \ge \NRD(R,n)$ but so far no asymptotic separation was known between these quantities. We exhibit an explicit arity $4$ relation for which $\CL(R,n) \ge \omega(\NRD(R,n))$. 
\end{abstract}

\section{Introduction}

Constraint satisfaction problems (CSPs) are foundational objects in  computer science, capturing the feasibility of many natural combinatorial questions. More formally, an instance of a constraint satisfaction problem is given by five things: (1) an arity $r$, specifying the number of variables that each constraint acts on, (2) the domain $D$, specifying the alphabet from which each variable is selected, (3) the $n$ variables, $x_1, \dots x_n$, each of which is free to be chosen from $D$, (4) the constraining \emph{relation} $R \subseteq D^r$, and (5) the $m$ constraints, specified by ordered subsets of $n$ of size $r$, denoted $S_1, \dots S_m$. Given such an instance of $\mathrm{CSP}(R)$, the goal is to find an assignment $x_1, \dots x_n \in D^n$ such that for every $i \in [m]$, $x_{S_i} \in R$, equivalently meaning the $i$th constraint is satisfied. Constraint satisfaction problems have a rich history of study, with seminal works seeking an understanding of when (as a function of $R$) satisfiability of a CSP instance is easy to determine \cite{FederV98,bulatov2005classifying, bulatov2017dichotomy, zhuk2020proof}, or how well they can be approximated \cite{GoemansW95,creignou2001complexity,khanna2001approximability,Hastad01,Khot02,KhotKMO07,Raghavendra08,MakarychevM17,KhotMS18,BrakensiekHPZ26} 
(in the sense of satisfying as many constraints as possible). 

Since these works, focus in the study of CSPs has transitioned to understanding other basic properties, like their sparsifiablity, kernelizability, space complexity of decidability, and learnability (which we will define and elaborate more on below). As we shall see, despite the fact that these properties are a priori unrelated, there are two closely related parameters that emerge in the modern study of \emph{all} of these different properties: non-redundancy (NRD) and chain-length (CL). 

Coined by Bessiere, Carbonnel, and Katisrelos~\cite{bessiere2020Chain}, non-redundancy and chain-length  both capture the extent to which different constraints in a CSP instances are independent (in the sense that they are not implied by the others). More formally, given a CSP instance with relation $R \subseteq D^r$, variables $x_1, \dots x_n$, and constraints $S_1, \dots S_m$, we say that the instance is \emph{non-redundant} if, for every constraint $i \in [m]$, there exists an assignment $x^{(i)} \in D^n$ such that $x^{(i)}|_{S_i} \notin R$, but  $x^{(i)}|_{S_j}  \in R$, for every $j \neq i$. In other words, the satisfaction of the $i$th constraint is \emph{not} implied by the satisfaction of the other constraints. Given a relation $R$ and a number of variables $n$, we use $\mathrm{NRD}(R, n)$ to denote the \emph{maximum number} of constraints in any non-redundant instance of a CSP with relation $R$. Chain-length is defined similarly, but with a small modification; given an instance with relation $R \subseteq D^r$, variables $x_1, \dots x_n$, and constraints $S_1, \dots S_m$, we say that the instance is \emph{a valid witness to chain length}, if for every constraint $i \in [m]$, there exists an assignment $x^{(i)} \in D^n$ such that $x^{(i)}|_{S_i} \notin R$, but  $x^{(i)}|_{S_j} \in R$, for every $j < i$. Just as in the definition of non-redundancy, we use $\CL(R,n )$ to denote the maximum number of constraints in any CSP instance with relation $R$ over $n$ variables which is a valid witness to chain-length. The key difference between the definitions of chain-length and non-redundancy is in the choice of $j$: in non-redundancy, we require that each constraint is \emph{uniquely unsatisfiable}, meaning that the $i$th constraint is unsatisfied, while all others are satisfied. Chain length instead asks that the $i$th constraint is unsatisfied, while only all \emph{previous} constraints are satisfied (with no restriction on constraints after the $i$'th constraint).

As mentioned above, chain-length and non-redundancy have emerged as central parameters governing the complexity of many properties of CSPs, which we  briefly discuss below.

\paragraph{Streaming Complexity of Decidability.} Given a relation $R$, one can ask what the space complexity is of determining whether a CSP instance is satisfiable, when the CSP instance is provided in \emph{a stream}. In this setting, each constraint is revealed one at a time and the goal is, at the end of the stream, to report whether the CSP instance defined by the combination of all constraints, is satisfiable while using as little space as possible. The recent work of Sharma and Velusamy \cite{SV26a} shows that this space complexity is tightly determined by its non-redundancy: with an $\Omega(\NRD(R, n))$ lower bound, and an $O(\NRD(R) \log(n))$ upper bound. 

\paragraph{Sparsifiability.} A separate line of work studies the \emph{sparsifiability} of CSPs~\cite{KK15,FK17,BZ20,KPS24,khanna2025efficient,brakensiek2025redundancy,BrakensiekGP26}. In this setting, one is given a CSP instance with constraints $S_1, \dots S_m$, weights $w_1, \dots w_m$, and a target accuracy parameter $\eps > 0$. The goal is to find a \emph{sparsifier}, which is a small subset of constraints $T \subseteq [m]$, along with new weights $w': T \rightarrow \R_{\geq 0}$ such that the weight of satisfied constraints on each assignment is approximately preserved; i.e., for every $x \in D^n$,
\[
\sum_{i \in [m]} w_i \cdot \mathbf{1}[x_{S_i} \in R] \in (1 \pm \eps) \sum_{i \in T} w'_i \cdot \mathbf{1}[x_{S_i} \in R].
\]
The sparsifiability is then measured as the \emph{smallest possible number} of retained constraints. The work of Brakensiek and Guruswami \cite{brakensiek2025redundancy} showed that for weighted CSPs (i.e., when the weights $w_1, \dots w_m$ can vary arbitrarily), the sparsifier size is governed by $\CL(\bar{R}, n)$ (see also, \cite{BrakensiekGP26}), and for unweighted CSPs (when $w_i = 1$), the sparsifier size is governed by $\NRD(\bar{R}, n)$. 

\paragraph{Learnability.} In the setting of learnability,  there is a hidden CSP instance with relation $R$ over $n$ variables. The goal is to learn the CSP instance (or more formally, the space of satisfying assignments) given access only via \emph{partial membership queries}. These queries specify values to a subset of the variables, and the response is a single bit indicating whether all constraints defined completely on the specified set of variables are satisfied. The work of Bessiere, Carbonnel, and Katisrelos \cite{bessiere2020Chain}, building on \cite{bessiere2013constraint}, showed that any learning algorithm for CSPs with relation $R$ requires $\Omega(\NRD(R, n))$ queries, and can be done in $O(\CL(R,n) \log(n))$ queries. 

\paragraph{Kernelizability.} The question of kernelizability asks whether any CSP instance with relation $R$ over $n$ variables can be efficiently reduced to a \emph{smaller} instance without altering the satisfiability. In this setting, \cite{jansen2019optimal,chen2020BestCase,lagerkvist2020Sparsification} showed that many instances can be efficiently kernelized to $O(\NRD(R, n))$ constraints. Proving such kernelization can occur for all CSP constraints is a major open problem posed by Carbonnel~\cite{carbonnel2022Redundancy}. See also recent work of Haviv~\cite{Haviv26} for discussion on the important link between non-redundancy and kernelization.

\subsection{Understanding Non-redundancy and Chain-length}

Motivated by the above applications, a flurry of recent work has sought to understand the non-redundancy and chain-length of various relations. For instance, the work of \cite{FK17,BZ20,bessiere2020Chain, lagerkvist2020Sparsification,carbonnel2022Redundancy,KPS24,khanna2025efficient} studied the non-redundancy (and also the chain-length) of relations expressible as linear equations and polynomials over groups, establishing a general principle that non-redundancy and chain-length are bounded by $O(n^c)$, where $c$ is the degree of a polynomial embedding of the relation, leading to a full classification of the non-redundancy and chain length of all arity $2$ CSPs~\cite{FK17,BZ20,bessiere2020Chain}. Later works extended this (in the setting of non-redundancy) even further, with \cite{brakensiek2025redundancy} demonstrating relations with non-integral non-redundancy (i.e., $n^{c+o(1)}$ for $c \notin \Z$) and \cite{brakensiek2025Richness} even demonstrating relations with non-redundancy $\Theta(n^{p/q})$ for any rational $p/q \geq 1$.

Even more recently, work has pivoted towards trying to understand the non-redundancy of \emph{all relations}. The work of \cite{khanna2025efficient} provided tight bounds for the non-redundancy of all Boolean relations of arity $3$. This was later extended 
by the work of \cite{brakensiek2026classification}, which provided tight bounds for all Boolean relations except one of arity $4$ ($R \subseteq \zo^4$). The work of \cite{SV26b} focused purely on the symmetric Boolean case, and classified all but two such relations for the case of arity $5$. The work of \cite{BGJLW26} has instead focused on developing an even more unified and comprehensive understanding of when a relation's non-redundancy is $\Theta(n)$.

\subsection{Our Work}\label{sec:overview}

However, despite this abundance of recent activity in studying non-redundancy and chain-length, even very basic questions remain open. Indeed, despite the utility of chain-length and non-redundancy, and the vast effort poured into understanding them, we still do not know if \emph{chain-length and non-redundancy are asymptotically the same}. That is to say, for every single relation whose non-redundancy and chain-length have been established to date, it is the case that $\NRD(R, n) = \Theta(\CL(R, n))$. 
This is for good reason too; many of the techniques used to bound non-redundancy (such as embedding into linear equations) also work directly for bounding the chain-length. This begs the central question of our work:

\begin{center}
    \emph{Are non-redundancy and chain-length always asymptotically the same?\\That is, for every $R$, is $\NRD(R,n) = \Theta(\CL(R, n))$?}
\end{center}

As our primary contribution, we answer the above in the negative:

\begin{theorem}\label{thm:main}
    There is a relation $R \subseteq D^4$, with $|D| = 3$ such that $\CL(R, n) \ge \Omega(n^3)$ but $\NRD(R,n) \le \frac{n^3}{2^{\Omega(\log^*(n))}}$.
\end{theorem}

In fact, we show that the non-redundancy of this relation is closely related to a problem in extremal combinatorics of packing hyperedges that avoid a specific configuration. More formally, let $f^{(3)}(n, 6, 3)$ denote the maximum number of hyperedges in a $3$-uniform hypergraph on $n$ vertices such that no $6$ vertices completely contain $3$ hyperedges. Our theorem in fact shows that $\NRD(R,n) \le O \left ( n \cdot  f^{(3)}(n, 6, 3)\right )$ (equivalently, this may be viewed as the largest number of edges in a graph which is the union of $n$ induced matchings -- the famous Rusza-Szemeredi problem). The upper bound of \cref{thm:main} then follows from the bound $f^{(3)}(n, 6, 3) \leq \frac{n^2}{2^{\Omega(\log^*(n))}}$, due to a famous work of Fox \cite{fox2011new}. However, the best lower bound for $f^{(3)}(n, 6, 3)$ is only known to be $\frac{n^2}{2^{O(\sqrt{\log(n)})}}$. Thus, it is distinctly possible that the separation between non-redundancy and chain-length for this relation is \emph{significantly larger} than $2^{\Omega(\log^*(n))}$.

Note that this connection between non-redundancy and $f^{(3)}(n, 6, 3)$-type bounds is not new, as several prior works \cite{bessiere2020Chain,brakensiek2025tight}, have made connections between non-redundancy and these extremal hypergraph problems (see also \cite{brakensiek2025Richness}). Our key contribution is showing that these $f^{(3)}(n, 6, 3)$-type upper bounds can in fact \emph{coexist} with distinctly longer chains in a single CSP instance.

\paragraph{The Separating Relation and an Overview of Techniques} We now briefly mention the intuition underlying \cref{thm:main}. The relation is constructed with the intention of having an \emph{explicit} chain of length $\Omega(n^3)$, while at the same time ensuring that the chain is as far from being a certificate of non-redundancy as possible. In more detail, we can recall that non-redundancy asks for a set of constraints where the $i$th constraint can be uniquely unsatisfied whereas chain-length asks only that the $i$th constraint is unsatisfied, while constraints preceding the $i$'th one are satisfied. The key difference is that the constraints corresponding to indices larger than $i$ are \emph{undetermined} and can be freely satisfied or unsatisfied in the definition of chain-length. Thus, our goal will construct a chain for which all constraints below the $i$'th constraint are \emph{still} unsatisfied when the $i$'th constraint is satisfied, thereby ensuring that the collection is as far as possible from being non-redundant. 

The primary question then becomes how to define such a relation. While there are many ways to do this, rather than immediately targeting an $\Omega(n^3)$ length chain, we instead make a simpler target of an $\Omega(n^2)$ length chain with a relation of arity $3$. To do this, we take our $n$ variables, and break them into three groups: $A, B, C$ where $|A| = n/4$, $|B| = n/4$, and $|C| = n/2$, assuming without loss of generality that $n$ is a multiple of $4$. The constraints are now formed by taking the $i$th variables from $A$, the $j$th variable from $B$, and the $(i + j)$'th variable from $C$. Thus, constraint $S_{i,j} = (a_i, b_j, c_{i +j})$. Because chain-length requires an ordering of the constraints, we then just order the constraints $S_{i,j}$ lexicographically by $i$ and then $j$.   Recall that the goal for the constraint $S_{i,j}$ is to find an assignment to the variables in $A, B,$ and $C$ such that $S_{i,j}$ is unsatisfied, while all preceding assignments are satisfied. We now define the domain of size three as $D = \{-, \z, +\}$. Our goal is to force the \emph{unsatisfying assignment} in our relation to always be $\z\z\z$, while also ensuring the extra condition that constraints before $S_{i, j}$ receive assignments from the relation
\[P = \{---, -\z-, -+-, -+\z, -++, \z--\}\]
when $S_{i,j}$ receives $\z\z\z$.

These choices of symbols may be mysterious at first glance, but appear for a natural reason. To build the ``witnessing assignment'' for $S_{i,j}$ (i.e., which doesn't satisfy $S_{i,j}$, and satisfies those preceding), we build the assignment $x$ such that:
\[
x(a_{i'}) = \mathrm{sgn}(i' -i) : i' \in A, \quad x(b_{j'}) = \mathrm{sgn}(j' -j) : j' \in B, \quad x(c_k) = \mathrm{sgn}(k - i - j): k \in C.
\]
As a sanity check, the constraint $S_{i,j} = (i, j, i+j)$ then receives exactly $\z\z\z$ as $x(a_i) = \mathrm{sgn}(i - i) = \z$, $x_{b_j} = \mathrm{sgn}(j - j) = \z$, and $x_{c_{i+j}} = \mathrm{sgn}(i + j - i - j) = \z$. We can then also see where the assignments in $P$ come from: these exactly capture \emph{all ways} that a constraint can come \emph{before} $S_{i,j}$ in the lexicographic ordering:
\begin{enumerate}
    \item If $S_{i', j'}$ is such that $i' < i, j' < j$, then $i'+j' < i + j$, and so the signs assigned by $x$ are $---$.
    \item If $S_{i', j'}$ is such that $i' < i$ but $j ' = j$, then $i' + j' < i + j$, and  so the signs assigned by $x$ are $-\z-$.
    \item If $S_{i', j'}$ is such that $i' < i$ but $j ' > j$, then $i' + j'$ can be bother larger, smaller, or equal to $i + j$. In this case the final sign is undefined, and so the signs assigned by $x$ can be any of $-+-, -+\z,$ or $-++$.
    \item  If $S_{i', j'}$ is such that $i' = i$ but $j ' < j$, then $i' + j' < i + j$, and so the signs assigned by $x$ are $\z--$.
\end{enumerate}

This instance we have defined satisfies a property even stronger than chain length: whenever $S_{i,j}$ is unsatisfied, it receives assignment $\z\z\z$, all preceding constraints receive an assignment from $P = \{---, -\z-, -+-, -+\z, -++, \z--\}$. We formalize such phenomenon as a bound on the chain length of the \emph{conditional predicate} $P \mid Q$, where $Q = P \cup \{\z\z\z\}$.

For this conditional predicate $P \mid Q$, one can ask a corresponding conditional non-redundancy question~\cite{bessiere2020Chain,brakensiek2025redundancy}. Precisely, how many constraints can we find which are \emph{non-redundant} and still satisfy the above split; meaning for constraint $S_{i,j}$, when $S_{i,j}$ is unsatisfied, it should receive $\z\z\z$, and when a constraint $S_{i,j}$ is satisfied, it receives an assignment from $P$. Understanding conditional non-redundancy in fact reveals an \emph{exact} connection to the $f^{(3)}(n, 6, 3)$ problem discussed above, meaning it must be asymptotically smaller than $n^2$, and therefore asymptotically smaller than the conditional chain-length (we omit a discussion of this equivalence here for brevity, see \cref{sec:CNRDUB} for more discussion).

To finish, we need to convert this conditional predicate into a ``standard relation'' $R \subseteq D^4$. We do this by adapting a trick introduced in \cite{brakensiek2025Richness} of setting
\[
    R = P \times \{\z\} \cup Q \times \{+\}.
\]
By adapting the techniques of \cite{brakensiek2025Richness,brakensiek2026classification,BGJLW26}, we show that the non-redundancy of $R$ is at most $O(n)$ times the non-redundancy of $P \mid Q$ while the chain length of $R$ is at least $\Omega(n)$ times the chain length of $P \mid Q$. As such, our super-constant separation between the non-redundancy and chain length of $P \mid Q$ is preserved in the relation $R$, albeit with an increase in the arity of the relation.

\subsection{AI Acknowledgment}

The separating predicate was found in conversation with ChatGPT 5.6-Sol Ultra. The authors had previously (without ChatGPT assistance) discovered relations for which the non-redundancy was bounded above by quantities exactly proportional to $f^{(3)}(n, 6, 3)$. ChatGPT was then used to perform an exhaustive search to find a relation where the same $f^{(3)}(n, 6, 3)$ NRD upper bound coincides with a existence of a long chain. Although aspects of the analysis of this predicate were suggested by ChatGPT (such as the polynomial appearing in \cref{claim:poly}), all text in this manuscript is human-written, and the authors are fully responsible for the entire contents of the manuscript.

\section{Preliminaries}

We now describe some facts and definitions concerning constraint satisfaction problems and extremal combinatorics.

\subsection{CSPs, $\NRD$, and $\CL$}

To start, we recall the notion of a constraint satisfaction problem:

\begin{definition}
    Given a \emph{relation} (or \emph{predicate}) $R \subseteq D^r$, we say that an \emph{instance} of $\CSP(R)$ is an ordered tuple $(V, C)$, where $V$ denote the \emph{variables} of the instance and $C \subseteq V^r$ denotes the \emph{constraints} (or \emph{clauses}) of the instance. An \emph{assignment} to the instance is a map $\varphi : V \to D$ for which we say that $S = (v_1, \hdots, v_r) \in C$ is \emph{satisfied} if $\varphi(S_i) := (\varphi(v_1), \hdots, \varphi(v_r)) \in R$.
\end{definition}

Typically, we let $n = |V|$ denote the number of variables and $m = |C|$ denote the number of constraints. Unless otherwise specified, we assume that $V = [n] := \{1, 2, \hdots, n\}$. We also often give an explicit enumeration $C = \{S_1, \hdots, S_m\}$ of the constraints. With the notion of a CSP established, we can now define what it means for an instance of a CSP to be non-redundant:

\begin{definition}[Non-Redundant CSP]
    Given a relation $R \subseteq D^r$, we say that an instance of $\CSP(R)$ with $n$ variables and $m$ constraints $S_1, \dots S_m$ is \emph{non-redundant} if for each constraint $S_i$, there exists an assignment $\varphi_i: [n] \rightarrow D$ such that:
    \begin{enumerate}
        \item For $j \neq i \in [m]$, $\varphi_i(S_j) \in R$.
        \item $\varphi_i(S_i) \notin R$.
    \end{enumerate}
\end{definition}

The non-redundancy of a relation $R$ on $n$ variables is simply the \emph{maximum size} of any non-redundant CSP instance:

\begin{definition}
    The non-redundancy of $R\subseteq D^r$ on $n$ variables is denoted $\NRD(R, n)$ and is defined as the maximum integer $m$ such that there exists a non-redundant instance of $\CSP(R)$ on $n$ variables and $m$ constraints.
\end{definition}

We similarly define chain-length:
\begin{definition}[Valid Witnesses to Chain Length]
    Given a relation $R \subseteq D^r$, we say that an instance of $\CSP(R)$ with $n$ variables and $m$ constraints $S_1, \dots S_m$ is \emph{a valid witness to chain length} if for each constraint $S_i$, there exists an assignment $\varphi_i: [n] \rightarrow D$ such that:
    \begin{enumerate}
        \item For $j < i \in [m]$, $\varphi_i(S_j) \in R$.
        \item $\varphi_i(S_i) \notin R$.
    \end{enumerate}
\end{definition}

The key distinction in the above is that each assignment $\varphi_i$ is only required to be satisfying for constraints that come \emph{before} $i$, as opposed to all constraints not equal to $i$. 
The chain-length of a relation $R$ on $n$ variables is simply the \emph{maximum size} of any CSP instance which is a valid witness to chain length:

\begin{definition}
    The chain-length of $R\subseteq D^r$ on $n$ variables is denoted $\CL(R, n)$ and is defined as the maximum integer $m$ such that there exists an instance of $\CSP(R)$ on $n$ variables and $m$ constraints which is a valid witness to chain length.
\end{definition}

Finally, we will require what we call \emph{conditional} versions of the above (cf. \cite{brakensiek2025redundancy}).

\begin{definition}[Conditional CSP]
Given relations $P \subsetneq Q \subseteq D^r$, we let $P \mid Q$ denote a \emph{conditional relation} (or \emph{conditional predicate}). An instance of $\CSP(P \mid Q)$ with $n$ variables and $m$ constraints is precisely an ordered tuple $(V, C)$ with $C \subseteq V^r$, $n = |V|$, and $m = |E|$. An \emph{assignment} to the instance is a map $\varphi : V \to D$ for which we say that $S = (v_1, \hdots, v_r) \in C$ is \emph{weakly satisfied} if $\varphi(S_i) \in Q$ and \emph{strongly satisfied} if $\varphi(S_i) \in P$.
\end{definition}

\begin{definition}[Non-redundancy for Conditional CSP]\label{def:conditionalNRD}
    Let $P \subsetneq Q \subseteq D^r$ be relations. An instance of $\CSP(P \mid Q)$ on $n$ variables and $m$ constraints $S_1, \dots S_m$ is \emph{non-redundant} if for each constraint $S_i$, there exists an assignment $\varphi_i: [n] \rightarrow D$ such that:
    \begin{enumerate}
        \item For $j \neq i \in [m]$, $\varphi_i(S_j) \in P$.
        \item $\varphi_i(S_i) \in Q \setminus P$.
    \end{enumerate}
    In other words, $\varphi_i$ weakly satisfies every constraint and further strongly satisfies all but $S_i$.
\end{definition}

\begin{definition}[Chain Length Witness for a Conditional CSP]\label{def:conditional-CL}
   Let $P \subsetneq Q \subseteq D^r$ be relations. An instance of $\CSP(P \mid Q)$ on $n$ variables and $m$ constraints $S_1, \dots S_m$ is a \emph{valid witness to chain length} if for each constraint $S_i$, there exists an assignment $\varphi_i: [n] \rightarrow D$ such that:
    \begin{enumerate}
        \item For $j < i \in [m]$, $\varphi_i(S_j) \in P$.
        \item $\varphi_i(S_i) \in Q \setminus P$.
    \end{enumerate}
    In particular, $\varphi_i$ weakly satisfies the constraints $\{S_1, \hdots, S_i\}$ and strongly satisfies the constraints $\{S_1, \hdots, S_{i-1}\}$.
\end{definition}

We use $\NRD(P \mid Q, n)$ and $\CL(P \mid Q, n)$ to denote the maximum size of a non-redundant instance and chain length witness for $\CSP(P \mid Q)$, respectively.

\begin{remark}
Our work appears to be the first to define chain length in the context of conditional predicates.
\end{remark}

\subsection{Extremal Combinatorics}

We will also require some basic notation from extremal combinatorics:

\begin{definition}
    We use $f^{(3)}(n, 6, 3)$ to denote the maximum number of hyperedges in a $3$-uniform hypergraph on $n$ vertices, such that no $6$ vertices completely contain $3$ hyperedges. 
\end{definition}

We have the following upper bound on $f^{(3)}(n, 6, 3)$:

\begin{fact}[See, for instance, \cite{fox2011new}]\label{fact:boundonRS}
    $f^{(3)}(n, 6, 3) \leq \frac{n^2}{2^{\Omega(\log^*(n))}}$.
\end{fact}

\section{The Separating (Conditional) Predicate}

In this section, we present a \emph{conditional} predicate $P \mid Q$ whose non-redundancy and chain length are asymptotically distinct. The relation will be of arity $3$, on a domain of size $3$. We represent the domain as $D = \{-, \z, +\}$, as this will be convenient for our later analysis. The relation $P \subseteq D^3$ will have $6$ satisfying tuples:
\begin{align}\label{eq:defOfPredicate}
    P = \{---, -\z-, -+-, -+\z, -++, \z--\}.
\end{align}
The relation $Q$ will simply be $P \cup \{\z\z\z\}$.

\subsection{Lower Bounding Chain-Length}

We now proceed to lower bounding the chain length of $\CSP(P \mid Q)$. For this, we have the following lemma:

\begin{lemma}\label{lem:CL-LB}
For $P, Q$ as defined in \cref{eq:defOfPredicate}, and any integer $n$, $\CL(P \mid Q, n) \ge \Omega(n^2)$.
\end{lemma}

We provide a formal proof below, but also point the reader to \cref{sec:overview} which contains an intuitive description of the lower bound.
\begin{proof}
    We start by describing the CSP instance which will certify the conditional chain-length. Given the universe of $n$ variables, we split the variables into groups $X, Y, Z$, where $|X| = \frac{n}{4}$, $|Y| = \frac{n}{4}$ and $|Z| = \frac{n}{2}$ (we assume WLOG that $n$ is a multiple of $4$). We will index $X$ and $Y$ with the numbers $\{0, 1, \dots n/4 -1\}$, and index $Z$ with $\{0, 1, \dots n/2-1\}$.

    Now, our CSP instance contains $\left ( \frac{n}{4}\right)^2$ many constraints, one for each choice of variables in $X, Y$. Indeed, for $x \in X$ and $y \in Y$, we include a constraint operating on variables $(x, y, (x+y))$. Importantly, note that $(x+y) \in \{0, 1, \dots n/2-1\}$, and thus we can associate it with a variable in $Z$. We denote this constraint by $S_{xy}$.

    Next, we must specify an ordering of the constraints. Indeed, we choose a trivial lexicographic ordering where $S_{xy}$ comes before $S_{x'y'}$ if $x < x'$ or $x = x'$ and $y < y'$.

    Now, it remains only to prove that the instance is a valid witness to chain length for $\CSP(P \mid Q)$. For this, let us consider the constraint $S_{xy}$. We define the assignment $\varphi_{xy}: X \cup Y \cup Z \rightarrow \{-, 0, +\}$ in accordance with the following:
    \begin{itemize}
        \item For $x' \in X$, $\varphi_{xy}(x') = \begin{cases}
            - \text{ if } x' < x \\
            \z \text{ if } x' = x \\
            + \text{ if } x' > x
        \end{cases}$
         \item For $y' \in Y$, $\varphi_{xy}(y') = \begin{cases}
            - \text{ if } y' < y \\
            \z \text{ if } y' = y \\
            + \text{ if } y' > y
        \end{cases}$
         \item For $z \in Z$, $\varphi_{xy}(z) = \begin{cases}
            - \text{ if } z < x + y \\
            \z \text{ if } z = x + y \\
            + \text{ if } z > x + y
        \end{cases}$
    \end{itemize}

    Importantly, we can then observe that for the constraint $S_{xy}$, $\varphi_{xy}(S_{xy})= \z\z\z$, as $x - x = 0, y - y = 0,$ and $(x+y) - (x+y) = 0$. Importantly, this then means that  $\varphi(S_{xy}) \in Q \setminus P$. Otherwise, we consider $S_{x'y'}$ for any constraint which \emph{precedes} $S_{xy}$. Thus, $S_{x'y'}$ falls in one of two cases:
    \begin{enumerate}
        \item The first possibility is that $x' < x$. In this case we see that $\varphi_{xy}(x') = -$. At the same time, $\varphi_{xy}(y')$ is unrestricted, it can be $-, \z$ or $+$. However, importantly, we know that the third variable in $S_{x'y'}$ is assigned 
        \[
        \varphi_{xy}(S_{x'y'} \cap Z) = \begin{cases}
            - \text{ if } x' + y' < x + y \\
            \z \text{ if } x' + y' = x + y \\
            + \text{ if } x' + y' > x + y
        \end{cases}.
        \]
        Crucially, there are then only three cases for us to understand, conditioned on the value of $\varphi_{xy}(y')$. If $\varphi_{xy}(y') = -$, then it must be the case that $x' + y' < x + y$, as both $x' < x$ and $y' < y$. So, then $\varphi_{xy}(S_{x'y'}) = --- \in P$. Otherwise, it is possible that $\varphi_{xy}(y') = \z$. This means that $y = y'$. Again then, it must be the case that $x' + y' < x + y$, as $x' < x$. So, in this case $\varphi_{xy}(S_{x'y'}) = -\z- \in P$. The final case is that $\varphi_{xy}(y') = +$. In this case, the relationship between $x' + y'$ and $x +y$ is undetermined. So, all we know is that $\varphi_{xy}(S_{x'y'}) \in \{-+-, -+\z, -++\}$. However, all of these tuples are in $P$! Thus, whenever $x' < x$, we see that it must be that  $\varphi_{xy}(S_{x'y'}) \in P$.
        \item The second possibility is that $x' = x$. In this case, we then know that it must be that $y' < y$, and therefore $x' + y' < x + y$. So, $\varphi_{xy}(S_{x'y'}) = \z-- \in P$.
    \end{enumerate}

    Thus, we see that this instance is a conditionally valid witness to chain length, as $\varphi_{xy}(S_{xy}) = 000 \in Q \setminus P$, and $\varphi_{xy}(S_{x'y'}) \in P$ whenever $x' < x$ or $x' = x$ and $y' < y$. Because the instance has $\Omega(n^2)$ constraints, we then obtain that $\CL(P \mid Q, n) \ge \Omega(n^2)$.
\end{proof}

\subsection{Upper Bounding Non-Redundancy}\label{sec:CNRDUB}

Next, it remains to \emph{upper bound} the conditional non-redundancy of this predicate. For this, we have the following lemma:

\begin{lemma}\label{lem:conditionalUpperBound}
    For $P, Q$ as defined in \cref{eq:defOfPredicate}, and any integer $n$, $\NRD(P \mid Q, n) \leq 27 f^{(3)}(n, 6, 3)$.
\end{lemma}

As mentioned in the introduction, using known bounds on $f^{(3)}(n, 6, 3)$ gives the following:

\begin{corollary}\label{cor:UBonCNRD}
$\NRD(P \mid Q, n) \leq \frac{n^2}{2^{\Omega(\log^*(n))}}$.
\end{corollary}

\begin{proof}
    This follows from plugging in the bound of \cref{fact:boundonRS} to \cref{lem:conditionalUpperBound}.
\end{proof}

\begin{proof}[Proof of \cref{lem:conditionalUpperBound}]
    We consider any CSP instance $n$ variables and $m$-constraints which is non-redundant for $P \mid Q$. Consider a uniformly random $3$-coloring of the vertices $c : [n] \to [3]$. For each constraint $S_i \in [n]^3$ for $i \in [m]$, the probability that for all $i \in [3]$, the $i$th variable of $S_i$ maps to $i$ is precisely $1/3^3 = 1/27$. Thus, there exists a subset of at least $m/27$ constraints for which the underlying hypergraph $H \subseteq [n]^3$ is $3$-partite. We let $X, Y, Z \subseteq [n]$ denote the three variable groups of the tripartition. To prove \cref{lem:conditionalUpperBound}, it suffices to prove that $|H| \le f^{(3)}(n, 6, 3)$.

    Now, as per \cref{def:conditionalNRD}, in order for the CSP instance to be non-redundant for $R \mid Q$, it must be the case that for every $S_i \in H$, there exists $\varphi_i: X \cup Y \cup Z \rightarrow \{-, 0, +\}$ such that $\varphi_i(S_i) = 000$ (as $000$ is the only tuple in $Q \setminus P$), and $\varphi_i(S_j) \in R$ (for $j \neq i$). We now proceed to upper bounding the number of hyperedges (constraints) in $H$. To start, we show that $H$ is a \emph{linear hypergraph}, meaning that any two sets intersect in at most $1$ vertex.

    \begin{claim}
    Let $H$ be defined as above. Then, for every $e, e' \in H$, it must be that $|e \cap e'| \leq 1$. 
    \end{claim}

    \begin{proof}
        Indeed, suppose for the sake of contradiction that there are some $e, e'$ for which $|e \cap e'| \geq 2$. Then, there must be an assignment $\varphi_e$ such that $\varphi_e(e) = \z\z\z$, while $\varphi_e(e') \in P$. However, $\varphi_e(e')$ contains at least two $\z$'s, and there is no assignment in $R$ which has at least two $\z$'s.
    \end{proof}

    Using the above, we will show something \emph{even stronger}; namely that there is no collection of $6$ vertices whose induced hypergraph contains $3$ hyperedges. Towards proving this, we assume for the sake of contradiction that some collection of such $6$ vertices $V' \subseteq X \cup Y \cup Z$ exists. We immediately see that $|V' \cap X| = |V' \cap Y| = |V' \cap Z|$; indeed, if any $|V' \cap A| = 0$ for $A \in \{X, Y, Z\}$, then no hyperedges will be in the induced subgraph of $V'$, and if $|V' \cap A| = 1$ for some $A \in \{X, Y, Z\}$, then all three hyperedges would share that single vertex in $A$. This gives a contradiction as this would imply that the hyperedges \emph{do not intersect} in any other vertices (otherwise it violates the linearity property established above). But this is not possible: there are only $5$ other vertices in $V'$, yet there are $6$ remaining vertices to be incident to the 3 hyperedges.

    So, as mentioned, we now focus on the case where $|V' \cap X| = |V' \cap Y| = |V' \cap Z| = 2$. We label the vertices in $V'$ as $x_1, x_2, y_1, y_2, z_1, z_2$. We assume WLOG that one of the hyperedges defined on $V'$ is $e_1 = (x_1, y_1, z_1)$, and that a second hyperedge uses the vertex $x_1$. In particular, this then forces the second hyperedge to include vertices $e_2 = (x_1, y_2, z_2)$ by the linearity conditions above. We now have a few possible cases based on the final hyperedge. Note that this final hyperedge must use $x_2$:
    \begin{enumerate}
        \item Suppose that the final hyperedge touches vertex $y_1$. Then it is forced that $e_3 = (x_2, y_1, z_2)$, as otherwise $e_3$ would be incident to two vertices from $(x_1, y_1, z_1)$ which violates linearity. However, this forces a contradiction with the conditional non-redundancy. We know there must be an assignment $\varphi$ which gives $(x_1, y_2, z_2)$ values $\z\z\z$, while the other hyperedges are in $R$. Because $e_1$ also uses $x_1$, this means that $\varphi(e_1) = \z--$. Because $e_2$ uses $z_2$, this also forces that $\varphi(e_2) = -+\z$; however, these are not simultaneously possible, as under a single assignment, this requires that $y_1$ gets values $-$ and $+$.
        \item Otherwise, suppose that the final hyperedge $e_3$ touches vertex $y_2$. Then it is forced by linearity of $H$ that $e_3 = (x_2, y_2, z_1)$. However, this again forces a contradiction with the conditional non-redundancy. Indeed, there must be an assignment $\varphi$ such that $\varphi(e_1) = \z\z\z$, while also ensuring that $\varphi(e_2), \varphi(e_3) \in R$. However, because $e_1$ and $e_2$ share $x_1$, this forces that $\varphi(e_2) = \z--$. Likewise, because $e_1$ and $e_3$ share $z_1$, this forces that $\varphi(e_3) = -+\z$. But, this means that $\varphi$ assigns $y_2$ to both $+$ and $-$ which is not possible. Thus, this violates the conditional non-redundancy. 
    \end{enumerate}
    Because the cases are exhaustive (the final hyperedge must touch one of $y_1, y_2$), this means that $H$ cannot contain any induced hypergraph on $6$ vertices with $3$ hyperedges. By the definition of $f^{(3)}(n, 6, 3)$, this yields the desired lemma. 
\end{proof}

\subsection{Unconditional separation}

So far we have demonstrated a separation of \emph{conditional} non-redundancy and chain length. We now show that this separation can be translated into an ``unconditional'' predicate with only a constant-factor change in the separation. The trick for making this translation was introduced by \cite{brakensiek2025Richness} and was subsequently adapted by \cite{brakensiek2026classification,BGJLW26}. In short, define $P \mid Q$ as in \cref{eq:defOfPredicate} and let
\[
    R := (P \times \{\z\}) \cup (Q \times \{+\}) \subseteq \{-, \z, +\}^4.
\]
We claim the following bounds which proves \cref{thm:main}.
\begin{theorem}\label{thm:main-formal}
For any positive integer $n$,
\begin{align}
\NRD(R, n) &\leq O(n \cdot f^{(3)}(n, 6, 3)) \le \frac{n^3}{2^{\Omega(\log^*(n))}}\label{eq:NRD}\\
\CL(R, n) &\geq \Omega(n^3).\label{eq:CL}
\end{align}
\end{theorem}

Although the methods of \cite{brakensiek2025Richness,brakensiek2026classification,BGJLW26} have not been extended to chain length, the same general strategies still apply, which we prove in detail below in \cref{subsec:CL}.

\subsubsection{Proof of \cref{eq:NRD}}\label{subsec:NRD}

To prove \cref{eq:NRD}, we first make use of an inequality due to \cite{brakensiek2026classification}.\footnote{They only state the inequality for arity $3$ predicates, but the same proof goes through nearly verbatim for any arity.}

\begin{fact}[Proposition 2.4~\cite{brakensiek2026classification}, adapted]
For any $P \subseteq Q \subseteq D^3$ and any distinct $a, b \in D$, we have that
\begin{align}
    \NRD((P \cup \{a\}) \cup (Q \cup \{b\})) \le \NRD(P \mid Q, n) \cdot O(n) + \NRD(Q \times \{0,1\}).\label{eq:R-bound}
\end{align}
\end{fact}
Using \cref{eq:R-bound} and \cref{lem:conditionalUpperBound}, we immediately get that
\[
    \NRD(R, n) \le O(n \cdot f^{(3)}(n, 6, 3)) + \NRD(Q \times \{\z,+\}, n)
\]
Since $n \cdot f^{(3)}(n, 6, 3) = n^{3-o(1)}$, it suffices to prove that $\NRD(Q \times \{\z,+\}, n) = O(n^2)$ to establish \cref{eq:NRD}. To do this, we use an established technique~\cite{chen2020BestCase, lagerkvist2020Sparsification, khanna2025efficient, brakensiek2025Richness, brakensiek2026classification,SV26b} of showing that $Q \times \{0,1\}$ corresponds to the locus of a quadratic polynomial (over some Abelian group). In particular, consider the quadratic polynomial $q : \R^4 \to \R$ defined by
\[
    q(x_1, x_2, x_3, x_4) = -x_1^2 + 2x_1x_2 - x_1x_3 + x_2x_3 + 10x_4^2 + 3x_2 - 2x_3 - 10x_4.
\]
Let $\sigma : \{-, \z, +\} \to \R$ be the map which sends $\sigma(-) = -1, \sigma(\z) = 0, \sigma(1) = +$, then we claim the following.

\begin{claim}\label{claim:poly}
For all $t \in \{-, \z, +\}^4$, we have that $q(\sigma(t)) = 0$ iff $t \in Q \times \{0,1\}$.
\end{claim}

\begin{proof}
Such a fact can be checked with brute-force computation. To give a bit more insight, let $p(x_1, x_2, x_3) = -x_1^2 + 2x_1x_2 - x_1x_3 + x_2x_3 + 3x_2 - 2x_3$ so that $q(x_1,  x_2, x_3, x_4) = p(x_1, x_2, x_3) + 10x_4(x_4-1)$. Note that for all $(x_1, x_2, x_3) \in \{-, \z, +\}^3$, we have that $|p(x_1, x_2, x_3)| \le 9$. Thus, $q(x_1,  x_2, x_3, x_4) = 0$ if and only if $p(x_1, x_2, x_3) = 0$ and $x_4 \in \{0, 1\}$. 

Thus, it suffices to verify that $p(\sigma(t)) = 0$ for $t \in \{-, \z, +\}^3$ if and only if $t \in Q$. That is, for $(x_1, x_2, x_3) \in \{-1, 0, 1\}^3$, we seek to understand when $p(x_1, x_2, x_3) = 0$. We break into cases.
\begin{itemize}
\item $(x_1, x_2) = (-1, -1)$. Note that $p(-1, -1, x_3) = -2 - 2x_3$ which is zero precisely when $(x_1, x_2, x_3) = (-1, -1, -1)$, which corresponds to a tuple of $Q$.
\item $(x_1, x_2) = (-1, 0)$. Note that $p(-1, 0, x_3) = -1 - x_3$ which is zero precisely when $(x_1, x_2, x_3) = (-1, 0, -1)$, which corresponds to a tuple of $Q$.
\item $(x_1, x_2) = (-1, 1)$. Note that $p(-1, 1, x_3) = 0$. Here, $(-1, 1, -1)$, $(-1, 1, 0)$, and $(-1, 1, 1)$ all correspond to tuples of $Q$.
\item $(x_1, x_2) = (0, -1)$. Note that $p(0, -1, x_3) = -3 - 3x_3$ which is zero precisely when $(x_1, x_2, x_3) = (0, -1, -1)$ which corresponds to a tuple of $Q$.
\item $(x_1, x_2) = (0, 0)$. Note that $p(0, 0, x_3) = -2x_3$ which is zero precisely when $(x_1, x_2, x_3) = (0, 0, 0)$ which corresponds to the last tuple of $Q$.
\item $(x_1, x_2) = (0, 1)$. Note that $p(0, 1, x_3) = 3-x_3$, which has no roots in $\{-1, 0, 1\}$.
\item $(x_1, x_2) = (1, -1)$. Note that $p(1, -1, x_3) = -6-4x_3$, which has no roots in $\{-1, 0, 1\}$.
\item $(x_1, x_2) = (1, 0)$. Note that $p(1, 0, x_3) = -1-3x_3$, which has no roots in $\{-1, 0, 1\}$.
\item $(x_1, x_2) = (1, 1)$. Note that $p(1, 1, x_3) = 4-2x_3$, which has no roots in $\{-1, 0, 1\}$.
\end{itemize}
This completes the proof.
\end{proof}

By these aforementioned works, \cref{claim:poly} (e.g., Theorem 10 of \cite{lagerkvist2020Sparsification}, see also Corollary 4.12 of \cite{chen2020BestCase}) shows that $\NRD(Q \times \{\z,+\}, n) \le O(n^2)$.

\subsubsection{Proof of \cref{eq:CL}}\label{subsec:CL}

To prove \cref{eq:CL}, it suffices to give an explicit chain of $\Omega(n^3)$ clauses on $n$ variables. Assume without loss of generality that $n$ is even. Let $m = \NRD(R, n/2)$ and recall by \cref{lem:CL-LB} that $m = \Omega(n^2)$. Let $V_1$ and $V_2$ each be sets of size $n/2$. Let $e_1, \hdots, e_m \in V_1^3$ be a set of $m$ edges found by \cref{lem:CL-LB} which form a valid witness to chain length of $P \mid Q$. That is, for all $i \in [m]$, there exists a map $f_i : V_1 \to \{-, \z, +\}$ such that $f_i(e_j) \in P$ for all $j < i$ and $f_i(e_i) \in Q \setminus P$.

Assume that $V_2 = \{x_1, \hdots, x_{n/2}\}$. For $i \in [mn/2]$, define $e'_i \in (V_1 \cup V_2)^4$ to be $(e_{j,1}, e_{j,2}, e_{j,3}, x_k)$ where $k = \lceil i / m\rceil$ and $j \in [m]$ such that $j \equiv i \mod m$. We claim that $e'_1, \hdots, e'_{mn/2}$ form a valid witness to the chain length of $R$. In particular, $\CL(R, n) \ge mn/2 = \Omega(n^3)$, which confirms \cref{eq:CL}.
 
To see why, for all $i \in [mn/2]$, with  $k = \lceil i / m\rceil$ and $j \in [m]$ such that $j \equiv i \mod m$, we define $g_i : V_1 \cup V_2 \to \{-, \z, +\}$ such that $g_i(v) = f_j(v)$ for all $v \in V_1$ and $g_i(x_\ell) = +$ if $\ell < i$ and $g_i(x_\ell) = \z$ otherwise. To see why this is a valid certificate, first observe that $g_i(e'_i) = (f_j(e_{j,1}), f_j(e_{j,2}), f_j(e_{j,3}), \z) \in (Q \setminus P) \times \{\z\} \subsetneq R$. Further, for any $i' < i$, let $k' = \lceil i'/m\rceil$ and $j' \in [m]$ such that $j' \equiv i' \mod m$. Then,
\[
    g_i(e'_{i'}) = (f_j(e_{j',1}), f_j(e_{j',2}), f_j(e_{j',3}), g_i(x_{k'})) \subseteq Q \times \{\z, +\}. 
\]
If $k' < k$ then $g_i(x_{k'}) = +$, so $g_i(e'_{i'}) \subseteq Q \times \{+\} \subseteq R$. Otherwise, if $k' = k$, then $j' < j$ so $(f_j(e_{j',1}), f_j(e_{j',2}), f_j(e_{j',3})) \in P$ due to the fact that $e_1, \hdots, e_m$ is a valid witness to chain length of $P \mid Q$.  Thus, $g_i(e'_{i'}) \subseteq P \times \{\z, +\} \subseteq R$. Thus, we have exhibited a valid witness to chain length of $R$ of cubic length.

\section{Conclusion}

In this paper, we gave a the first example of a CSP predicate for which its non-redundancy and chain length have a super-constant multiplicative separation. We hope that this work spurs further research into finding possible separations between non-redundancy and chain length. In particular, we conjecture that there exists predicates which have \emph{polynomial} separation between chain length and non-redundancy (see a similar prediction in \cite{brakensiek2025redundancy}).

\begin{conjecture}\label{conj:separation}
There exists a finite domain $D$, an arity $r$, and non-trivial predicate $R \subseteq D^r$ such that there exists a constant $\eps > 0$ for which
\[
\frac{\CL(R, n)}{\NRD(R, n)} = \Omega(n^{\eps}).
\]
\end{conjecture}

Building on \cref{thm:main}, one possible route to proving \cref{conj:separation} is to find an ``asymmetric'' hypergraph Tur{\'a}n problem whose asymmetric structure can be exploited to build an unexpectedly long chain. Sadly, due to the dearth of known bounds for degenerate hypergraph Tur{\'a}n problems (e.g., \cite{MR776819}), novel extremal combinatorics techniques may be needed to execute such an approach.

\bibliographystyle{alpha}
\bibliography{references}

\newcommand{\etalchar}[1]{$^{#1}$}
\begin{thebibliography}{KKMO07}

\bibitem[BCH{\etalchar{+}}13]{bessiere2013constraint}
Christian Bessiere, Remi Coletta, Emmanuel Hebrard, George Katsirelos, Nadjib
  Lazaar, Nina Narodytska, Claude{-}Guy Quimper, and Toby Walsh.
\newblock Constraint acquisition via partial queries.
\newblock In Francesca Rossi, editor, {\em {IJCAI} 2013, Proceedings of the
  23rd International Joint Conference on Artificial Intelligence, Beijing,
  China, August 3-9, 2013}, pages 475--481. {IJCAI/AAAI}, 2013.

\bibitem[BCK20]{bessiere2020Chain}
Christian Bessiere, Cl{\'e}ment Carbonnel, and George Katsirelos.
\newblock Chain {{Length}} and {{CSPs Learnable}} with {{Few Queries}}.
\newblock {\em Proceedings of the AAAI Conference on Artificial Intelligence},
  34(02):1420--1427, April 2020.

\bibitem[BG25]{brakensiek2025redundancy}
Joshua Brakensiek and Venkatesan Guruswami.
\newblock Redundancy is all you need.
\newblock In Michal Kouck{\'{y}} and Nikhil Bansal, editors, {\em Proceedings
  of the 57th Annual {ACM} Symposium on Theory of Computing, {STOC} 2025,
  Prague, Czechia, June 23-27, 2025}, pages 1614--1625. {ACM}, 2025.

\bibitem[BGJ{\etalchar{+}}25]{brakensiek2025Richness}
Joshua Brakensiek, Venkatesan Guruswami, Bart~MP Jansen, Victor Lagerkvist, and
  Magnus Wahlstr{\"o}m.
\newblock The richness of {CSP} non-redundancy.
\newblock {\em arXiv preprint arXiv:2507.07942}, 2025.

\bibitem[BGJ{\etalchar{+}}26]{BGJLW26}
Joshua Brakensiek, Venkatesan Guruswami, Bart~MP Jansen, Victor Lagerkvist, and
  Magnus Wahlstr{\"o}m.
\newblock Super-linear lower bounds for csp non-redundancy via shrinking
  instances.
\newblock {\em arXiv preprint arXiv:2605.19055}, 2026.

\bibitem[BGP25]{brakensiek2025tight}
Joshua Brakensiek, Venkatesan Guruswami, and Aaron Putterman.
\newblock Tight bounds for sparsifying random csps.
\newblock {\em arXiv preprint arXiv:2508.13345}, 2025.

\bibitem[BGP26a]{brakensiek2026classification}
Joshua Brakensiek, Venkatesan Guruswami, and Aaron Putterman.
\newblock Classification of non-redundancy of boolean predicates of arity 4.
\newblock In Nicolas Beldiceanu, editor, {\em 32nd International Conference on
  Principles and Practice of Constraint Programming, {CP} 2026, Lisbon,
  Portugal, July 20-23, 2026}, volume 379 of {\em LIPIcs}, pages 8:1--8:24.
  Schloss Dagstuhl - Leibniz-Zentrum f{\"{u}}r Informatik, 2026.

\bibitem[BGP26b]{BrakensiekGP26}
Joshua Brakensiek, Venkatesan Guruswami, and Aaron Putterman.
\newblock Multiplicative error set system sparsification: {A} simpler proof via
  chain length contraction.
\newblock In Sayan Bhattacharya, Danupon Nanongkai, Michael Benedikt, and
  Gabriele Puppis, editors, {\em 53rd International Colloquium on Automata,
  Languages, and Programming, {ICALP} 2026, Royal Holloway, University of
  London, Egham, United Kingdom, July 7-10, 2026}, volume 374 of {\em LIPIcs},
  pages 44:1--44:17. Schloss Dagstuhl - Leibniz-Zentrum f{\"{u}}r Informatik,
  2026.

\bibitem[BHPZ26]{BrakensiekHPZ26}
Joshua Brakensiek, Neng Huang, Aaron Potechin, and Uri Zwick.
\newblock Separating {MAX} 2-and, {MAX} di-cut, and {MAX} {CUT}.
\newblock {\em {SIAM} J. Comput.}, 55(3):S23--268, 2026.

\bibitem[BJK05]{bulatov2005classifying}
Andrei Bulatov, Peter Jeavons, and Andrei Krokhin.
\newblock Classifying the complexity of constraints using finite algebras.
\newblock {\em SIAM journal on computing}, 34(3):720--742, 2005.

\bibitem[Bul17]{bulatov2017dichotomy}
Andrei~A Bulatov.
\newblock A dichotomy theorem for nonuniform {CSP}s.
\newblock In {\em 2017 IEEE 58th Annual Symposium on Foundations of Computer
  Science (FOCS)}, pages 319--330. IEEE, 2017.

\bibitem[B{\v Z}20]{BZ20}
Silvia Butti and Stanislav {\v Z}ivn{\'y}.
\newblock Sparsification of {{Binary CSPs}}.
\newblock {\em SIAM Journal on Discrete Mathematics}, 34(1):825--842, January
  2020.

\bibitem[Car22]{carbonnel2022Redundancy}
Cl{\'e}ment Carbonnel.
\newblock On {{Redundancy}} in {{Constraint Satisfaction Problems}}.
\newblock In {\em {{DROPS-IDN}}/v2/Document/10.4230/{{LIPIcs}}.{{CP}}.2022.11}.
  Schloss Dagstuhl -- Leibniz-Zentrum f{\"u}r Informatik, 2022.

\bibitem[CJP20]{chen2020BestCase}
Hubie Chen, Bart M.~P. Jansen, and Astrid Pieterse.
\newblock Best-{{Case}} and {{Worst-Case Sparsifiability}} of {{Boolean CSPs}}.
\newblock {\em Algorithmica}, 82(8):2200--2242, August 2020.

\bibitem[CKS01]{creignou2001complexity}
Nadia Creignou, Sanjeev Khanna, and Madhu Sudan.
\newblock {\em Complexity classifications of Boolean constraint satisfaction
  problems}.
\newblock SIAM, 2001.

\bibitem[FK17]{FK17}
Arnold Filtser and Robert Krauthgamer.
\newblock Sparsification of two-variable valued constraint satisfaction
  problems.
\newblock {\em {SIAM} J. Discret. Math.}, 31(2):1263--1276, 2017.

\bibitem[Fox11]{fox2011new}
Jacob Fox.
\newblock A new proof of the graph removal lemma.
\newblock {\em Annals of Mathematics}, pages 561--579, 2011.

\bibitem[FV98]{FederV98}
Tom{\'{a}}s Feder and Moshe~Y. Vardi.
\newblock The computational structure of monotone monadic {SNP} and constraint
  satisfaction: {A} study through datalog and group theory.
\newblock {\em {SIAM} J. Comput.}, 28(1):57--104, 1998.

\bibitem[GW95]{GoemansW95}
Michel~X. Goemans and David~P. Williamson.
\newblock Improved approximation algorithms for maximum cut and satisfiability
  problems using semidefinite programming.
\newblock {\em J. {ACM}}, 42(6):1115--1145, 1995.

\bibitem[H{\aa}s01]{Hastad01}
Johan H{\aa}stad.
\newblock Some optimal inapproximability results.
\newblock {\em J. {ACM}}, 48(4):798--859, 2001.

\bibitem[Hav26]{Haviv26}
Ishay Haviv.
\newblock Kernelization bounds for constrained coloring.
\newblock In Michal Kouck{\'{y}} and Daniela Petrisan, editors, {\em 51st
  International Symposium on Mathematical Foundations of Computer Science,
  {MFCS} 2026, Paris, France, August 24-28, 2026}, volume 386 of {\em LIPIcs},
  pages 53:1--53:14. Schloss Dagstuhl - Leibniz-Zentrum f{\"{u}}r Informatik,
  2026.

\bibitem[JP19]{jansen2019optimal}
Bart~MP Jansen and Astrid Pieterse.
\newblock Optimal sparsification for some binary {CSP}s using low-degree
  polynomials.
\newblock {\em ACM Transactions on Computation Theory (TOCT)}, 11(4):1--26,
  2019.

\bibitem[Kho02]{Khot02}
Subhash Khot.
\newblock On the power of unique 2-prover 1-round games.
\newblock In {\em Proceedings of the 17th Annual {IEEE} Conference on
  Computational Complexity, Montr{\'{e}}al, Qu{\'{e}}bec, Canada, May 21-24,
  2002}, page~25. {IEEE} Computer Society, 2002.

\bibitem[KK15]{KK15}
Dmitry Kogan and Robert Krauthgamer.
\newblock Sketching cuts in graphs and hypergraphs.
\newblock In Tim Roughgarden, editor, {\em Proceedings of the 2015 Conference
  on Innovations in Theoretical Computer Science, {ITCS} 2015, Rehovot, Israel,
  January 11-13, 2015}, pages 367--376. {ACM}, 2015.

\bibitem[KKMO07]{KhotKMO07}
Subhash Khot, Guy Kindler, Elchanan Mossel, and Ryan O'Donnell.
\newblock Optimal inapproximability results for {MAX-CUT} and other 2-variable
  csps?
\newblock {\em {SIAM} J. Comput.}, 37(1):319--357, 2007.

\bibitem[KMS18]{KhotMS18}
Subhash Khot, Dor Minzer, and Muli Safra.
\newblock Pseudorandom sets in grassmann graph have near-perfect expansion.
\newblock In Mikkel Thorup, editor, {\em 59th {IEEE} Annual Symposium on
  Foundations of Computer Science, {FOCS} 2018, Paris, France, October 7-9,
  2018}, pages 592--601. {IEEE} Computer Society, 2018.

\bibitem[KPS24]{KPS24}
Sanjeev Khanna, Aaron~(Louie) Putterman, and Madhu Sudan.
\newblock Code sparsification and its applications.
\newblock In David~P. Woodruff, editor, {\em Proceedings of the 2024 {ACM-SIAM}
  Symposium on Discrete Algorithms, {SODA} 2024, Alexandria, VA, USA, January
  7-10, 2024}, pages 5145--5168. {SIAM}, 2024.

\bibitem[KPS25]{khanna2025efficient}
Sanjeev Khanna, Aaron Putterman, and Madhu Sudan.
\newblock Efficient algorithms and new characterizations for {CSP}
  sparsification.
\newblock In Michal Kouck{\'{y}} and Nikhil Bansal, editors, {\em Proceedings
  of the 57th Annual {ACM} Symposium on Theory of Computing, {STOC} 2025,
  Prague, Czechia, June 23-27, 2025}, pages 407--416. {ACM}, 2025.

\bibitem[KSTW01]{khanna2001approximability}
Sanjeev Khanna, Madhu Sudan, Luca Trevisan, and David~P Williamson.
\newblock The approximability of constraint satisfaction problems.
\newblock {\em SIAM Journal on Computing}, 30(6):1863--1920, 2001.

\bibitem[LW20]{lagerkvist2020Sparsification}
Victor Lagerkvist and Magnus Wahlstr{\"o}m.
\newblock Sparsification of {{SAT}} and {{CSP Problems}} via {{Tractable
  Extensions}}.
\newblock {\em ACM Transactions on Computation Theory}, 12(2):1--29, June 2020.

\bibitem[MM17]{MakarychevM17}
Konstantin Makarychev and Yury Makarychev.
\newblock Approximation algorithms for csps.
\newblock In Andrei~A. Krokhin and Stanislav Zivn{\'{y}}, editors, {\em The
  Constraint Satisfaction Problem: Complexity and Approximability}, volume~7 of
  {\em Dagstuhl Follow-Ups}, pages 287--325. Schloss Dagstuhl - Leibniz-Zentrum
  f{\"{u}}r Informatik, 2017.

\bibitem[Rag08]{Raghavendra08}
Prasad Raghavendra.
\newblock Optimal algorithms and inapproximability results for every csp?
\newblock In Cynthia Dwork, editor, {\em Proceedings of the 40th Annual {ACM}
  Symposium on Theory of Computing, Victoria, British Columbia, Canada, May
  17-20, 2008}, pages 245--254. {ACM}, 2008.

\bibitem[Sim84]{MR776819}
Mikl\'{o}s Simonovits.
\newblock Extremal graph problems, degenerate extremal problems, and
  supersaturated graphs.
\newblock In {\em Progress in graph theory ({W}aterloo, {O}nt., 1982)}, pages
  419--437. Academic Press, Toronto, ON, 1984.

\bibitem[SV26a]{SV26a}
Amatya Sharma and Santhoshini Velusamy.
\newblock Characterizing streaming decidability of csps via non-redundancy,
  2026.

\bibitem[SV26b]{SV26b}
Amatya Sharma and Santhoshini Velusamy.
\newblock Non-redundancy of low-arity symmetric boolean csps, 2026.

\bibitem[Zhu20]{zhuk2020proof}
Dmitriy Zhuk.
\newblock A proof of the csp dichotomy conjecture.
\newblock {\em Journal of the ACM (JACM)}, 67(5):1--78, 2020.

\end{thebibliography}

\end{document}